\documentclass[11pt]{amsart}

\usepackage{amscd,amssymb,amsthm,verbatim}
\usepackage{enumerate}
\usepackage{bm}
\usepackage[margin=1.4in]{geometry}
\usepackage{mathrsfs, graphicx} 
\usepackage{stmaryrd}

\usepackage[%
bookmarks=true,
colorlinks,
linkcolor=blue,
urlcolor=blue,
citecolor=blue,
plainpages=false,
pdfpagelabels,
final,
breaklinks=true\between
]{hyperref}
\usepackage{hyperref}
\usepackage[numbers,sort&compress]{natbib}

\numberwithin{equation}{section}

\theoremstyle{plain}
\newtheorem{thm}{Theorem}[section]
\newtheorem{prop}[thm]{Proposition}
\newtheorem{lem}[thm]{Lemma}
\newtheorem{cor}[thm]{Corollary}

\theoremstyle{definition}
\newtheorem{defn}[thm]{Definition}
\theoremstyle{remark}
\newtheorem{rk}[thm]{Remark}
\DeclareFontFamily{U}{MnSymbolC}{}
\DeclareSymbolFont{MnSyC}{U}{MnSymbolC}{m}{n}
\DeclareFontShape{U}{MnSymbolC}{m}{n}{
	<-6>  MnSymbolC5
	<6-7>  MnSymbolC6
	<7-8>  MnSymbolC7
	<8-9>  MnSymbolC8
	<9-10> MnSymbolC9
	<10-12> MnSymbolC10
	<12->   MnSymbolC12}{}
\DeclareMathSymbol{\intprod}{\mathbin}{MnSyC}{'270}

\begin{document}

\title[Observations and predictions from past lightcones]{Observations and predictions from past lightcones}
\author{Martin Lesourd} 
\address{Black Hole Initiative, Harvard University}
\email{mlesourd@fas.harvard.edu}

\begin{abstract}
In a general Lorentzian manifold $M$, the past lightcone of a point is a proper subset of $M$ that does not carry enough information to determine the rest of $M$. That said, if $M$ is a globally hyperbolic Cauchy development of vacuum initial data on a Cauchy surface $S$ and there is a point whose past lightcone contains $S$, then the contents of such a lightcone determines all of $M$ (up to isometry). We show some results that describe what properties of $M$ guarantee that past lightcones do indeed determine all or at least significant portions of $M$. Null lines and observer horizons, which are well known features of the de-Sitter spacetime, play a prominent role.
\end{abstract}
\maketitle

\section{Introduction}
\hspace{0.2in} In Lorentzian geometry, an observer at a given time in a spacetime $(M,g)$ is represented by a timelike curve with future endpoint $p\in M$, and the past lightcone $J^-(p)\subset M$ of $p$ represents all signals in $M$ that can reach the observer at $p$. One can then ask the following. 
\begin{enumerate}
\item[\textbf{(A)}] Can an observer at $p\in M$ know the global structure of $M$ on the basis of $J^-(p)$?
\item[\textbf{(B)}] Can an observer at $p\in M$ make predictions about $M\backslash J^-(p)$ on the basis of $J^-(p)$?
\end{enumerate}
\hspace{0.2in} In this short note, we describe some of what is known about \textbf{(A)} and \textbf{(B)}, we prove various further results, and we list some natural further questions. \\ \indent 
Part of the appeal in \textbf{(A)} and \textbf{(B)} is that they are subject to somewhat surprising examples. As described below, two (inextendible and globally hyperbolic) spacetimes $(M',g')$ and $(M,g)$ can be non-isometric, in spite of the fact that each member of the countable collection of past lightcones $\{I^-(p_i)\}$ that covers $M$ can be isometrically embedded into $(M',g')$ and likewise with $M$ and $M'$ interchanged. Here we will show when this cannot happen. \\ \indent 
We now recall some basic definitions of causal theory, cf.\@ \cite{B}, \cite{MS}, \cite{ON} for some classic references and \cite{Minguzzisurvery} for a more recent authorative survey.\\ \indent 
A \textit{spacetime} $(M,g)$ is a connected \(C^{\infty}\) Hausdorff manifold $M$ of dimension two or greater with a Lorentzian metric \(g\) of signature \((-,+,+,...)\), and we will assume a time and space orientation. Since regularity is not the issue here, for simplicity of expression we take $g$ to be smooth, but many of the arguments can be extended to lower regularity.     \\ \indent
The lightcone structure inherited on the tangent space $T_pM$ at each $p$ leads to the notion of a causal curve, which in turn leads to defining $J^-(p)$ (or $I^-(p)$) as the collection of all points $q\in M$ from which there exists a causal (or timelike) curve with future endpoint $p$ and past endpoint $q$. A causal curve between $p$ and $q\in J^+(p)$ is \textit{achronal} iff $q\notin I^+(p)$. A \textit{null line} is an inextendible achronal causal curve.  \\ \indent
A set is \textit{achronal} (\textit{acausal}) iff no two members of it can be connected by a timelike (causal) curve. The \textit{domain of dependence} $D(S)$ of a set $S\subset M$ is given by $D(S)=D^+(S)\cup D^-(S)$ and $D^+(S)$ is defined as the collection of all points $q$ in $M$ such that any inextendible ($C^1$) past directed causal curve passing through $q$ intersects $S$. $\tilde{D}(S)$ is defined identically except that curves are timelike, rather than causal. From the perspective of the Cauchy problem of general relativity, $D^+(S)$ represents the maximal portion of $M$ that could be determined by initial data on $S$. If $S$ is closed as a subset of $M$ then $\overline{D^+(S)}=\tilde{D}^+(S)$. The \textit{future Cauchy horizon} is defined as $H^+(S)\equiv \overline{D^+(S)}\backslash I^-(D^+(S))$ and $H(S)=H^+(S)\cup H^-(S)$. \\ \indent 
$S$ is a \textit{partial Cauchy hypersurface} if it an edgeless acausal set, cf.\@ Definition 14.27 of \cite{B} for the definition of $\text{edge}(S)$ for an achronal set $S$. A spacetime $(M,g)$ is \textit{globally hyperbolic} iff there exists a partial Cauchy hypersurface $S$ such that $M=D(S)$. By \cite{BS03}, a smooth globally hyperbolic spacetime $(M,g)$ is isometric to $(\mathbb{R}\times S,-f(t)dt^2+h(t))$ where $f(t)$ is smooth and $h(t)$ a Riemannian metric on $S$. From the perspective of causal structure, global hyperbolicity is equivalent to the spacetime being \textit{causal}\footnote{No closed causal curves.} and $J^+(p)\cap J^-(q)$ being compact for all $p,q\in M$. \\ \indent 
If $S$ is acausal and $S\cap \text{edge}(S)=\emptyset$ (eg., $S$ is a partial Cauchy hypersurface), then $D(S)$ is non-empty, open, and $D(S) \cap H(S)=\emptyset$. Note that the openness of $D(S)$ may be ruined if we take $S$ to be achronal rather than acausal.\\ \indent 
A spacetime $(M,g)$ is \textit{causally simple} iff it is causal and $J^{+(-)}(p)$ is closed for all $p\in M$. Global hyperbolicity is strictly stronger than causal simplicity. \\ \indent 
An isometric embedding of $(M,g)$ into $(M',g')$ is an injective (but not necessarily surjective) map $\phi:M\hookrightarrow M'$ such that $\phi$ is a diffeomorphism onto its image $\phi(M)$ and $\phi^*(g')=g$. If $\phi$ maps $M$ surjectively onto $M'$ then we write $\phi:M\to M'$ and we say that $(M,g)$ and $(M',g')$ are \textit{isometric}. A spacetime $(M,g)$ is \textit{inextendible} when there exists no isometric embedding $\phi$ of $M$ into $M'$ such that $M'\backslash \phi(M)\neq \emptyset$. \\ \indent 
A spacetime is \textit{future holed} if there is a partial Cauchy hypersurface $S'\subset M'$ and an isometric embedding $ \phi:\tilde{D}(S')\hookrightarrow M$ such that $\phi(S')$ is acausal in $(M,g)$ where $(M,g)$ is any spacetime, and $\phi(H^+(S'))\cap D^+(\phi(S')) \neq \emptyset$. A spacetime is \textit{hole-free} if it lacks future and past holes. Minguzzi \cite{Ming1} shows that causally simple, inextendible spacetimes are hole free. \\  

\textbf{Acknowledgements}. We thank the Gordon Betty Moore Foundation and John Templeton Foundation for their support of Harvard's Black Hole Initiative. We thank Professors Erik Curiel, JB Manchak,  Ettore Minguzzi, Chris Timpson, and James Weatherall for valuable comments that improved the paper. 

\vspace{0.2in}
\section{Previous Work}
\textbf{(A)}. A natural definition coming from \cite{Mal} (whose terminology is slightly different) is the following.
\begin{defn}
A spacetime $(M,g)$ is \textit{weakly observationally indistinguishable} from $(M',g')$ just in case there is an isometric embedding $\phi: I^-(p) \hookrightarrow M'$ for every point $p\in M$. If this is also true with $M$ and $M'$ interchanged, then the spacetimes are \textit{strongly observationally indistinguishable}.
\end{defn}
One could replace $I^-(p)$ with $J^-(p)$, and although this makes no real difference to any of the results in \cite{Mal} and \cite{Manchak1} or indeed to what follows, that would capture a somewhat more honest sense of distinguishability because observers can certainly receive signals from $J^-(p)\backslash I^-(p)$.   \\ \indent
One can also take observers to be inextendible timelike curves $\sigma$ (as is also done in \cite{Mal}), but in that case $J^-(\sigma)\subset M$ would represent ``all possible observations that could be made supposing that the observer lasts forever with respect to $M$'', which, though interesting, is stronger than what happens in practice. \\  \indent 
Malament pg.\@ 65-6 \cite{Mal} gives examples of spacetimes that are strongly observationally indistinguishable but non-isometric. His examples are globally hyperbolic, inextendible, and exploit the presence of observer horizons in the de-Sitter spacetime. \\ \indent
Based on an explicit cut and paste argument, Manchak \cite{Manchak1} shows the following.
\begin{prop}[Manchak \cite{Manchak1}]\label{nemesis}
Given any non-causally bizarre\footnote{A spacetime is causally bizarre if there is a point $p\in M$ such that $M\subset I^-(p)$.} spacetime $(M,g)$, there exists a spacetime $(M',g')$ that is weakly observationally indistinguishable from $(M,g)$ but not isometric to $(M,g)$.
\end{prop}
Although Manchak's construction of $(M',g')$ works for any (non-causally bizarre) spacetime $(M,g)$, it relies on introducing a countably infinite collection of holes in $(M',g')$. It is unknown to us whether Proposition \ref{nemesis} holds for strong observational indistinguishability.  \\\ \\ \indent 
Viewed together, \cite{Mal} and \cite{Manchak1} lead one to the following. \\ \\
\textbf{Question.} Find conditions $\{ A,B,...\}$ satisfied by $(M,g)$ and $(M',g')$ such that: \\ \\
\hspace*{0.7in} `Weakly (or strongly) observationally indistinguishable $+ A + B +...$' \\ 
\hspace*{2.8in} $\Leftrightarrow$ \\ 
\hspace*{1.8in} `$(M,g)$ and $(M',g')$ are isometric' \\ \\ 
Proposition \ref{Nulllines} and Corollary \ref{deSitter} below are in this direction. \\ \\ 
\indent \textbf{(B)}. Geroch \cite{Ger} defines prediction in general relativity as follows.
\begin{defn}
$p\in M$ is \textit{predictable} from $q$, written as $p\in P(q)$, iff $p\in D^+(S)$ for some closed, achronal set $S\subset J^-(q)$. $p\in M$ is \textit{verifiable} iff $p\in P(q)$ and $p\in I^+(q)\backslash J^-(q)$.
\end{defn} 
Manchak observes the following.  
\begin{prop}[Manchak \cite{Manchak2}]
If $P(q)\cap (I^+(q)\backslash J^-(q))\neq \emptyset$, then $(M,g)$ admits an edgless compact achronal set.
\end{prop}
A slightly different notion of prediction considered in \cite{Manchak2} is as follows. 
\begin{defn}
$p\in M$ is \textit{genuinely predictable} from \(q\), written $\mathcal{P}(q)$, iff $p\in P(q)$ and for all inextendible spacetimes $(M',g')$, if there is an isometric embedding $\phi:J^-(q)\hookrightarrow M'$, then there is an isometric embedding $\phi':J^-(q)\cup J^-(p)\hookrightarrow M'$ such that $\phi =\phi'_{\mid J^-(q)}$.
\end{defn}
The idea here is that genuine predictions guarantee that observers with the same past make the same predictions. By a short cut-and-paste argument, Manchak observes the following. 
\begin{prop}[Manchak \cite{Manchak2}]\label{notgenuine}
Let $(M,g)$ be any spacetime and $q$ a point in $M$. Then $\mathcal{P}(q)\subseteq \partial J^-(q)$.
\end{prop}
Thus the domain of genuine predictions from \(q\), if non-empty, is on the verge of being a retrodiction.

\section{Some Observations}
\hspace{0.15in} We assume that spacetimes satisfy the field equations
\begin{equation}\label{Einstein}
G_{g}\equiv \text{Ric}_{g}-\frac{1}{2}g\: \text{Scal}_{g} =T_{\phi,...}-\Lambda \: g
\end{equation}
where $\Lambda\in \mathbb{R}$ is a constant, and with $T_{\phi,...}$ the stress-energy tensor associated with possible matter fields $\{ \phi,...\}$ in $M$.\footnote{Here we think of the Cauchy problem from the perspective of `initial data', as opposed to `initial and boundary data', though the latter is more natural for $\Lambda<0$.} The system \eqref{Einstein} leads to the formulation of a Cauchy problem on a spacelike initial data set $S$. In the vacuum setting $T=0$, $\Lambda=0$, the Cauchy problem was shown \cite{CBG} to be well posed in the sense that there exists a unique, up to isometry, maximal globally hyperbolic development of $S$ obtained by Cauchy evolution of the initial data on $S$ according to \eqref{Einstein} with $T=0$, $\Lambda=0$. This well posedness has been extended to more general settings and we take it as a pre-condition for the spacetimes we consider.
\begin{defn}
Fix the constant $\Lambda$ and fix an expression for $T_{\phi,...}$. Given a spacetime $(M,g)$ and an acausal edgeless connected set $S\subset M$, we say that $(\tilde{D}(S),g|_{\tilde{D}(S)})$ is a \textit{faithful development} if it is uniquely determined, up to isometry, by Cauchy evolution of the initial data on $S$ according to \eqref{Einstein}. If $(M,g)$ admits a connected acausal edgeless set, we say that it is \textit{locally Cauchy} if, for any connected acausal edgeless set $S$, $(\tilde{D}(S),g|_{\tilde{D}(S)})$ is a faithful development.
\end{defn} 
\begin{rk}
Note that this is slightly unorthodox in the sense that it is usually $D(S)$, rather than $\tilde{D}(S)$, which we think of as being determined by $S$. Given that $S$ is closed for the definition of locally Cauchy, we have $\overline{D(S)}=\tilde{D}(S)$, and so being locally Cauchy is only slightly stronger than asking for $D(S)$ to be determined up to isometry.
\end{rk}
\begin{rk}
Since we want to guarantee isometric embeddings, we want to rule out examples of regions that are globally hyperbolic but not determined by Cauchy evolution. To see a trivial example\footnote{Examples like this suggest the following problem. Given a Lorentzian manifold $(M,g)$ with an arbitrary geodesically complete Lorentzian metric $g$, what conditions on $M$ and $g$ make it possible to solve for a function $\Omega:M\to \mathbb{R}$ such that $(M,\Omega^2g)$ is complete and vacuum?} smooth functions transformations $\Omega$ , start with Minkowski spacetime $(\mathbb{R}^{1,n},\eta)$, identify some open set $O\subset \mathbb{R}^{1,n}$ lying above $t=0$, and modify it by a conformal transformation $\eta\to \Omega^{2}\eta$. In that case, in spite of $M\backslash O$ being vacuum, $O$ will in general have a non-vanishing Einstein tensor $G_{g}$. But since there are Cauchy hypersurfaces (in $M\backslash O$) for $(\mathbb{R}^{1,n},\Omega^{2}\eta)$ which are exactly flat\footnote{In the language of the constraints, initial data sets of the form ($\mathbb{R}^n,g_E,0$).}, $(\mathbb{R}^{1,n},\Omega^{2}\eta)$ is not locally Cauchy if it is not isometric to $(\mathbb{R}^{1,n},\eta)$.
\end{rk}
We also make use of the following.
\begin{defn}
Given a partial Cauchy hypersurface $S$ in a spacetime $M$ we say that an open subset of $M$ is an $\epsilon$-development of $S$, denoted $D_{\epsilon}(S)(\supset S)$ if there is exists an $\epsilon>0 \:(\in \mathbb{R})$ such that $D_{\epsilon}(S)$ admits a Cauchy surface $S_{\epsilon}$ every point of which lies at distance $\geq \epsilon>0$ in the future of $S$, as measured by a normalized timelike vector field orthogonal to $S$.
\end{defn}
\begin{rk}
The non-empty interior of a causal diamond $J^+(p)\cap J^-(q)\neq \emptyset$ with $p,q\in (\mathbb{R}^{1,n},\eta)$ is not an $\epsilon$-development because its Cauchy surfaces are anchored at $\partial J^+(p)\cap \partial J^+(q)$. 
\end{rk} 
We now observe the following. 
\begin{prop}\label{Nulllines}
Let $(M,g)$ and $(M',g')$ be inextendible and locally Cauchy. Suppose that
\begin{enumerate}
	\item[(i)] $(M,g)$ has a compact Cauchy surface and no null lines, 
	\item[(ii)] $(M',g')$ is causal and hole-free.
\end{enumerate}
Then $(M,g)$ and $(M',g')$ are isometric iff they are weakly observationally indistinguishable.
\end{prop}
Note that by the examples in \cite{Mal}, Proposition \ref{Nulllines} is false without the assumption that $(M,g)$ lacks null lines. \\ \indent 
In either case of weakly o.i. or strongly o.i., it would be interesting to settle whether the compactness in (i) is necessary, cf.\@ Proposition \ref{noncompact} below.
\begin{proof}
The proof of Proposition \ref{Nulllines} starts by strengthening Theorem 1 of \cite{GW}.\footnote{In \cite{GW}, the authors assume that $(M,g)$ is null geodesically complete, satisfies the null energy condition, and the null generic condition. These assumptions implies the absence of null lines. It was then observed by Galloway that one can prove the theorem by instead assuming that $(M,g)$ lacks null lines (cf.\@ footnote for Theorem 1 of \cite{GW}), but since those details never appeared we include them for completeness.}
\begin{lem}\label{GaoWald}
Let \((M,g)\) be a spacetime without null lines. Then given any compact region \(K\), there exists another compact \(K'\supset K\) such that if \(p,q\notin K'\) and \(q\in J^+(p)-I^+(p)\) , then any causal curve \(\gamma\) connecting \(p\) to \(q\) cannot intersect \(K\). 
\end{lem} 
\begin{proof}
Suppose otherwise for some \(K_0\) and \(K_1\supset K_0\). Then consider a sequence of ever bigger compact sets \(K_{i+1}\subset K'_{i}\). By assumption, each \(K_i\) will have horismos\footnote{The future horismos of $p$ is defined as $J^+\backslash I^+(p)$.} related outer points \(p_i,q_i\notin K_{i}\), \(q_i\in J^-(p_i)-I^-(p_i)\) that are connected by a causal curve, necessarily achronal, that intersects \(K_0\). In considering larger compact sets, we make these causal curves longer in the sense of an auxiliary Riemannian metric. All of these curves intersect \(K_0\). Taking the limit, by compactness of \(K_0\), there is a limit point for the sequence of points lying in \(K_0\) for each causal curve linking \(p_i,q_i\), which moreover is in \(K_0\). By a standard limit curve arguments, cf.\@ Proposition 3.1 of \cite{B}, there passes an inextendible limit curve through this limit point. The limit curve is also straightforwardly seen to be achronal.\footnote{See \cite{MLC} for significantly stronger limit curve statements.} Thus we have a null line. 
\end{proof} 
With the same proof as in Corollary 1 of \cite{GW}, Lemma \ref{GaoWald} implies the following.\footnote{Lemma \ref{horizons} strengthens the main result of \cite{GW}.}
\begin{lem}\label{horizons}
Let \((M,g)\) be a spacetime with compact Cauchy surface that does not admit any null lines. Then \((M,g)\) admits a point \(p\in M\) such that \(S\subset I^-(p)\) for some Cauchy surface \(S\). 
\end{lem} 
\begin{proof}
We include this for completeness, but this argument is exactly as in \cite{GW} except that Lemma \ref{GaoWald} plays the role of Theorem 1 of \cite{GW}. Since $(M,g)$ is globally hyperbolic, there exists a continuous global time function $t:M\to \mathbb{R}$, such that each surface of constant $t$ is a Cauchy surface. Let $K=\Sigma$ and $K'$ be as in Lemma \ref{GaoWald}. Let $t_1$ and $t_2$ denote, respectively, the minmum and maximum values of $t$ on $K'$. Let $\Sigma_1$ be any Cauchy surface with $t<t_1$ and let $\Sigma_2$ denote the Cauchy surface $t=t_2$. Let $q\in I^+(\Sigma_2)$, $p\in \Sigma_1$ and suppose that $p\in \partial I^-(q)$. Since $(M,g)$ is globally hyperbolic, $J^-(q)$ is closed and so $p\in J^-(q)\backslash I^-(q)$ and thus there is a causal curve connecting $p$ and $q$. It follows from Lemma \ref{GaoWald} that this causal curve does not intersect $\Sigma$. However, this contradicts the fact that $\Sigma$ is a Cauchy surface. Consequently, there cannot exist a $p\in \Sigma_1$ and such that $p\in \partial I^-(q)$, i.e., $\partial I^-(q)\cap \Sigma_1=\emptyset$. But $I^-(q)$ is open and since $\partial I^-(q)\cap \Sigma_1=\emptyset$, the complement of $I^-(q)$ in $\Sigma_1$ is also open. Since we have $I^-(q)\cap \Sigma_1$ and $\Sigma_1$ is connected, this implies $\Sigma_1\subset I^-(q)$.
\end{proof}
We now finish the proof of Proposition \ref{Nulllines}. By Lemma \ref{horizons} there is a point $p\in M$ with $S\subset I^-(p)$ where $S$ is a Cauchy surface of $M$. \\ \\ \indent 
\textbf{$\phi(S)$ is compact}. By weak observational indistinguishability, there is an isometric embedding $\phi:I^-(p)\hookrightarrow M'$. Because $\phi$ is a diffeomorphism onto its image, the map $\phi|_S$, $\phi$ restricted to $S$, is a diffeomorphism of $S$ onto its image and thus $\phi(S)$ is compact in $M'$. \\ \\ \indent 
\textbf{$\phi(S)$ is achronal.} Suppose otherwise that $\gamma'$ is a past directed timelike curve from $x'$ to $y'$ with $x,y\in \phi(S)$. Extend $\gamma'$ to $\sigma'$ so that $\sigma'$ is a past directed timelike curve from $q'$ to $x'$ to $y'$. Note that the isometric embedding forces $\phi(S)$ to be locally achronal; that is, there is an open neighborhood $O$ around around $S$ with two connected boundary components $\partial_+O(\subset I^+(S))$, $\partial_- O$ (allocated using the orientation in $M$), such that no two distinct points in $O$ can be joined by a timelike curve in $O$. Since $O$ isometrically embeds into $M'$, a locally achronal neighborhood $O'$ exists around $\phi(S)$ in $M'$. As such, the curve $\sigma'$ must leave $\partial_-O'$ and re-enter $O'$ via either $\partial_-O'$ or $\partial_+O'$. \\ \indent 
In the former case, we can build a closed piecewise smooth timelike curve from $q'$ and back, which violates causality of $M'$. \\ \indent 
In the latter case, we will obtain a contradiction with the inextendibility of $M$. Since $I^-(S)\subset D^-(S)\subset I^-(p)$, we must have that $\sigma'$ leaves $\phi(D^-(S))$, say at some point $r'\in \partial \phi(D^-(S))$, and re-enter $\phi(I^-(p))\cap I^+(\phi(S))$. Since the global hyperbolicity of $M$ implies that every future directed timelike curve in $\phi(I^-(p))$ must eventually leave $\phi(I^-(p))$ when sufficiently extended in the future direction, we can take the re-entry point to lie on the `future' boundary of $\phi(I^-(p))$; that is, which is the endpoint of a timelike curve whose $\phi^{-1}$ pre-image has endpoint on $\partial I^-(p)$. Now consider an open neighborhood $Z'$ of $r'$ in $M'\cap \partial \phi(D^-(S))$. Consider the open subset $\phi(I^-(p))\cup Z'$ of $M'$. Now define a new spacetime $M''$ by $\phi(I^-(p))\cup Z'\cup J^+(\partial I^-(p))$. We know this can be done because the global hyperbolicity of $M$ and the locally Cauchy property of $M$ and $M'$ imply that the regions $\overline{I^-(p)}$ and $\overline{\phi(I^-(p))}\backslash \partial \phi(D^-(S))$ are isometric. We now have a spacetime $M''$ into which $M$ can be isometrically embedded as a proper subset (in virtue of the extra $Z'$ beyond $\phi(D^-(S))$), contradicting the inextendibility of $M$. \\ \\ \indent
\textbf{$\phi(S)$ is edgeless.} Compactness of $S$ and achronality of $\phi(S)$ straightforwardly implies that $\phi(S)$ is edgeless. 
\begin{rk}
At this point if $(M',g')$ is assumed globally hyperbolic, $\phi(S)$ being an edgeless compact connected achronal set means that $\phi(S)$ can be taken to be a Cauchy surface of $M'$. In that case, both $(M,g)$ and $(M',g')$ are representatives of the unique, up to isometry, maximal globally hyperbolic development of $S$, and are thus isometric. We will instead show that $(M',g')$ is globally hyperbolic. 
\end{rk}
In $(M,g)$ we can consider an $\epsilon$-development $D_{\epsilon}(S)\subset M$. Within $D_{\epsilon}(S)$, we can then find an acausal hypersurface $S_{\epsilon}$ which is still Cauchy for $M$. We know that $D_{\epsilon}(S)\subset M$ isometrically embeds in $M'$ as a small neighborhood around $\phi(S)$. The image $\phi(S_\epsilon)$ of $S_\epsilon$, now denoted $S'_{\epsilon}$, is acausal in $M'$ (by a causal version of the argument for the achronality of $\phi(S)$) and edgeless. 
We now have partial Cauchy surfaces $S_{\epsilon}$ and $S'_{\epsilon}$ in $M$ and $M'$ respectively. \\ \\ \indent 
\textbf{We have the inclusion $D^-(S'_{\epsilon})\supseteq M'\backslash I^+(S'_{\epsilon})$}. Since $(M,g)$ is globally hyperbolic and inextendible, we have $D^-(S_{\epsilon})\supseteq J^-(S_{\epsilon})$. Since $J^-(S_{\epsilon})$ isometrically embeds into $(M',g')$, if $D^-(S'_{\epsilon})$ fails to cover $M'\backslash I^+(S'_{\epsilon})$, then $S'_{\epsilon}$ must have a past Cauchy horizon $H^-(S'_{\epsilon})\neq \emptyset $ in $M'$. In that case, by the locally Cauchy property we can isometrically embed $S'_{\epsilon}$ and $\tilde{D}^-(S'_{\epsilon})$ back into $M$ using $\phi^{-1}$, and by the global hyperbolicity and inextendibility of $M$, we have that $D^-(\phi^{-1}(S'_{\epsilon}))\supset \phi^{-1}(H^-(S'_{\epsilon}))$, which contradicts the past hole-freeness of $(M',g')$. \\ \\ \indent 
\textbf{We have the inclusion $D^+(S'_{\epsilon})\supseteq M'\backslash I^-(S'_{\epsilon})$.} This follows by the same argument.\\ \\ \indent 
Since we now have $D(S_\epsilon)=M$ and $D(S'_\epsilon)=M'$, the conclusion follows from locally Cauchy. 
\end{proof}
In view of the role of null lines, we note the rigidity theorem of Galloway-Solis \cite{GS}.
\begin{thm}[Galloway-Solis \cite{GS}]\label{GallowaySolis}
Assume that the $4$-dimensional spacetime $(M^{4},g)$ 
\begin{enumerate} 
	\item[(i)] satisfies \eqref{Einstein} with $T=0$\footnote{Generalizations to Einstein-Maxwell are possible.} and $\Lambda >0$,
	\item[(ii)] is asymptotically de-Sitter\footnote{cf. \cite{GS} for definitions of asymptotically de-Sitter and the associated hypersurfaces $\mathcal{J}^{+(-)}$ in that context.},
	\item[(iii)] is globally hyperbolic, 
	\item[(iv)] there is a null line with endpoints on $\mathcal{J}^+$ and $\mathcal{J}^-$. 
\end{enumerate}
Then $(M^4,g)$ isometrically embeds as an open subset of the de-Sitter spacetime containing a Cauchy surface. 
\end{thm}
Together, Proposition \ref{Nulllines} and Theorem \ref{GallowaySolis} imply the following.
\begin{cor}\label{deSitter}
Given two $4$-dimensional spacetimes $(M,g)$ and $(M',g')$ assume that
\begin{enumerate} 
	\item[(i)] $(M,g)$ and $(M',g')$ satisfy \eqref{Einstein} with $T=0$ and $\Lambda >0$,
	\item[(ii)] $(M,g)$ is inextendible and has a compact Cauchy surface,
	\item[(iii)] $(M,g)$ is asymptotically de-Sitter but not isometric to de-Sitter,
	\item[(iv)] $(M',g')$ is inextendible, causal and hole-free.
\end{enumerate}
Then $(M,g)$ and $(M',g')$ are isometric iff they are weakly observationally indistinguishable. 
\end{cor}
Note that the assumptions of Corollary \ref{deSitter} just falls short of astrophysical relevance on account of the assumption that $(M,g)$ be \textbf{past} asymptotically de-Sitter, which is not supported by current data. A more desirable statement would be welcome.\footnote{\cite{Ming2} contains results precluding the existence of null lines based on astrophysically interesting assumptions.} \\ \\
\textbf{(B)}. We say that a spacetime \((M,g)\) is \textit{Cauchy friendly} if it is weakly locally Cauchy\footnote{Weakly locally Cauchy replaces $\tilde{D}(S)$ with $D(S)$.} and there are no points \(p\in M\) such that \(J^-(p)\supseteq M\). We now show the following, which guarantees that genuine predictions extend a little beyond what is suggested in Proposition \ref{notgenuine}.
\begin{prop}\label{genuine}
Given two Cauchy friendly spacetimes \((M,g)\) and \((M',g')\), assume that
\begin{enumerate}
\item[(i)] there is a partial Cauchy surface $S\subset J^-(q)\subseteq D(S)$ for some point $q\in M$,
\item[(ii)] there is an isometry $\phi:J^-(q)\to J^-(q')$.
\end{enumerate}
Then there is an isometric embedding \(\psi:A \hookrightarrow M'\) for some $A\supsetneq J^-(q)$ such that 
\begin{itemize}
\item{\(\psi_{\mid J^-(q)}=\phi\),}
\item{\(A\) and $\psi(A)$ contain points in the domain of verifiable prediction of \(q\) and $q'$.}
\end{itemize}
\end{prop}
\begin{proof}
First we make some basic observations, and throughout we denote $S'\equiv \phi(S)$. \\ \\ \indent 
\textbf{We have $J^-(q)\subsetneq D(S)$}. Since \(J^-(q)\) lies in a globally hyperbolic set \(D(S)\), \(J^-(q)\) is closed. Moreover, since \(S\) is acausal and edgeless, \(D(S)\) must be open. From this it follows that the inclusion of \(J^-(q)\subseteq D(S)\) is strict. Suppose otherwise that \(J^-(q)=D(S)\). In that case \(D(S)\) is both open and closed in \(M\), and since \(M\) is connected, that implies \(D(S)=M\). But then \(M=J^-(q)\), in contradiction with Cauchy friendly. Thus \(J^-(q)\) is a closed proper subset of the open set \(D(S)\). \\ \\ \indent
\textbf{$S'$ is a partial Cauchy hypersurface.} Since $\phi$ is a diffeomorphism onto its image, we know that $S'$ is compact. Unlike the arguments given in Proposition \ref{Nulllines}, $\phi(S)$ is acausal and edgeless by the fact that $\phi$ is an isometry (as opposed to merely an embedding). Since $S'$ belongs to $J^-(q')$, if $S'$ is acausal in $M'$, then $S$ is also acausal in $M$, which contradicts (ii). Thus $S'$ is a partial Cauchy surface in $M'$.\\ \\ \indent 
\textbf{We have $J^-(x')\cap J^+(S')\subseteq D^+(S')$ for any $x'\in \phi(J^-(q))$}. We seek to show that any past inextendible causal curve in \(J^-(x')\) with with future endpoint \(x'\in \phi(J^-(q))\) intersects $S'$. By assumption, we have an isometric embedding $\phi:J^-(q)\cap D^+(S)\hookrightarrow M'$ where $S$ is a closed achronal set. Consider any point $x\in J^-(q)\cap J^+(S)$. By the isometry $\phi$, we have $\phi(J^-(x))= J^-(x')$. Note first that by definition, there is a neighborhood of \(0\in T_xM\) such that the exponential map of the past non-spacelike vectors in that neighborhood is contained in \(D^+(S)\). By the isometry $\phi$, the same is true for $x'$ with respect to \(\phi(D^+(S))\), in particular there is a neighborhood \(U_{x'}\) of \(x'\) such that \(U_{x'}\cap J^-(x')\subset \phi (D^+(S))\). Seeking a contradiction, suppose there is a past inextendible causal curve \(\gamma'\) with future endpoint \(x'\) that does not intersect $S'$. By the property aforementioned, there is at least a segment of $\gamma'$ contained in $\phi(D^+(S))$. The curve defined by \(\gamma\equiv \phi^{-1}(\gamma')\) is causal and ends at \(x\), and is thus entirely contained in \(D^+(S)\). But then \(\gamma\) does not intersect \(S\), which is a contradiction. \\ \\ \indent 
\textbf{Similarly, we have $J^-(S)\subseteq D^-(S')$.} This follows from (ii), the isometry $\phi$, and the argument just above. \\ \\ \indent 
\textbf{We have $J^-(q')\subsetneq D(S')$.} This proceeds as above, which now leads to a contradiction with the Cauchy friendliness of $M'$. \\ \\ \indent 
We have two closed sets \(J^-(q')\) and \(J^-(q)\) each strictly contained in the open sets \(D(S')\) and \(D(S)\). We now seek to show the existence of \(\psi\). Although there may be no isometric embedding of $D(S)$ into $M'$, we need only show that it is possible to non-trivially extend the pre-image of $\phi$ beyond \(J^-(q)\), which will thus enter $D(S)\backslash J^-(q)$.\\ \indent
Consider now the unique (up to isometry) maximal globally hyperbolic development \(X(S)\) of \(S\), where $X(S)$ denotes one representative among all isometric developments. By the isometry $\phi$, we know that $S$ and and $S'$ are isometric as initial data sets, and thus that \(X(S)\) is the unique (up to isometry) maximal globally hyperbolic development of both \(S\) and \(S'\). By locally Cauchy, it follows that both \(D(S)\) and \(D(S')\) can be isometrically embedded into \(X(S)\).\footnote{Note here that we could use a weaker version of locally Cauchy here that involves $D(S)$ rather than $\tilde{D}(S)$.} Denote these isometries by \(\rho:D(S)\hookrightarrow X(S)\) and \(\rho':D(S')\hookrightarrow X(S)\). \\ \indent 
It is obvious that \(\rho(D(S))\cap \rho'(D(S'))\neq \emptyset\) and since \(\rho\) and \(\rho'\) are local diffeomorphisms, both \(\rho(D(S))\) and \(\rho(D(S'))\) are open in $X(S)$, and moreover both \(\rho(J^-(q))\) and $\rho'(J^-(q')$ are closed in \(X(S)\). We now define the following set in $X(S)$
\[[\rho(D(S))-\rho(J^-(q))] \cap [\rho'(D(S'))-\rho(J^-(q'))]\equiv I\]
The openness of $\rho(D(S)), \rho'(D(S'))$, and the fact that we may choose $\rho,\rho'$ such that $\rho(J^-(q))=\rho'(J^-(q'))$ means that the strict inclusions $D(S)\backslash J^-(q),D(S')\backslash J^-(q')$ and $\neq \emptyset$ extend to $\rho$ and $\rho'$, i.e., $\rho(D(S))\backslash \rho(J^-(q))\neq \emptyset$ and $\rho(D(S'))\backslash \rho(J^-(q'))\neq \emptyset$. It follows that $I\neq \emptyset$. \\ \indent 
We can now identify the set \(A \equiv \rho^{-1}[\rho(J^-(q))\cup I]\) as having the desired properties, i.e.\@ there exists an isometric embedding of \(\psi:A\hookrightarrow M'\) such that $\psi|_{J^-(q)}=\phi$.
\end{proof}

We can also consider what happens after lifting the compactness assumption on $S$.
\begin{prop}\label{noncompact}
Let $S$ be a partial Cauchy hypersurface in $M$ and $D_{\epsilon}(S)$ an $\epsilon$-development of $S$. Let $\phi:D_{\epsilon}(S)\hookrightarrow M'$ be an isometric embedding into a hole-free spacetime $M'$. Then either
\begin{itemize}
\item{$\phi(S)$ is causal in $M'$,}
\item{or $\phi(S)$ is a partial Cauchy hypersurface in $M'$.}
\end{itemize} 
In the latter case, if $M,M'$ are locally Cauchy and $M'$ is inextendible, then there is an isometric embedding $\psi:D(S)\hookrightarrow M'$ with $\psi|_{D_{\epsilon}(S)}=\phi$.
\end{prop}
Thus, after basic assumptions like hole-freeness, locally Cauchy and inextendibility, the only obstruction concerns the acausality of $\phi(S)$. It may be that $\phi(S)$ is always acausal if $M'$ satisfies some causality assumption, eg.\@ causally simple, causally continuous\footnote{cf.\@ pg.\@ 59 of \cite{B}}, etc.\\ \indent
Note also that $\phi(S)$ need not be a partial Cauchy hypersurface if we replace $D_{\epsilon}(S)$ by `an open globally hyperbolic subset of $D(S)$' (delete a half-space at $t=0$ from $(\mathbb{R}^{1,n},\eta)$). 
\begin{proof} 
First we show the second statement. If $\phi(S)$ is a partial Cauchy hypersurface, we know that $D(S)$ and $D(\phi(S))$ are both open subsets of $M$ and $M'$ respectively, and since $M,M'$ are locally Cauchy, we know that $D(S)$ and $D(\phi(S))$ share a common (isometric) subset extending beyond $D_{\epsilon}(S)$. Let $\mathcal{D}(S)$ denote the maximal open globally hyperbolic subset of $D(S)$ for which there is an isometric embedding $\psi':\mathcal{D}(S)\hookrightarrow M'$. \\ \indent 
Now suppose that $D(S)$ does not isometrically embed into $M'$, i.e.\@ $\mathcal{D}(S)\subsetneq D(S)$. If $H(\psi'(S))\neq \emptyset$ then by locally Cauchy we can use $\psi'^{-1}$ to embed $\tilde{D}(\psi'(S))$ into $M$ and contradict the hole-freeness of $M'$. If $H(\psi'(S))= \emptyset$, then $M'=D(\psi'(S))$. But in that case, by locally Cauchy, $M'$ isometrically embeds into $M$ as a proper subset of $M$, contradicting the inextendibility of $M'$.\\ \indent   
Now we show the first part: if $\phi(S)$ is acausal, then it is a partial Cauchy hypersurface in $M'$. Supposing that $\text{edge}(\phi(S))\neq \emptyset$, we will show that $\text{edge}({\phi(S)})\cap \phi(S)=\emptyset $, and we then show that this implies $(M',g')$ is holed. \\ \indent 
Let $q'$ be a point in $\text{edge}(\phi(S))\cap \phi(S)$. In that case denote $q=\phi^{-1}(q')\in S$ and take a future directed timelike curve $\sigma$ from $q$ to $q_{\epsilon}\in I^+(S)\cap D_{\epsilon}(S)$, and set $\sigma$ to be past inextendible in $M$. Then there is a timelike curve $\sigma'=\phi(\sigma)\subset M'$ passing through $q'$. Take $U_i(q')$ to be a system of increasingly small neighborhoods $U_i(q')\supsetneq U_{i+1}(q')$, each containing points in $I^+(q')$ and $I^-(q')$, such that $\{U_i(q')\}$ has accumulation point $q'$. Define a collection of curves $\{\gamma_i'\}$ by taking $\sigma'$, removing from $\sigma'$ the portion $\sigma'\cap U_i(q')$, and replacing that portion with timelike segments with endpoints in $I^+(q')$ and $I^-(q')$ which miss $\phi(S)$. Although this might produce only piecewise smooth timelike curves, the curves $\{\gamma_i'\}$ can be approximated by $C^1$ causal curves (still missing $\phi(S)$), which we relabel as $\{\gamma_i'\}$. Now consider $\{\phi^{-1}(\gamma_i')\}$. This defines a collection of $C^1$ causal curves in $D_{\epsilon}(S)$ that approach $\sigma$ but which do not intersect $S$. \\ \indent 
We now recall a well known fact. In any spacetime $L$, there exists a sufficiently small neighborhood $N(x)\subset L$ around some point $x$ such that $J^+(y)\cap J^-(z)$ is compact for all $y$, $z$ $\in N(x)$. Since a sufficiently small $N(x)$ is causal, every point in a spacetime lives in a small globally hyperbolic neighborhood. \\ \indent 
Consider such a globally hyperbolic neighborhood $N(q)\subset M$ centered at $q$. Without loss of generality, we can take $N(q)$ to have Cauchy surface $S\cap N(q)$. For some sufficiently large $n\in \mathbb{N}$, the causal curves $\{\phi^{-1}(\gamma_{i\geq n}')\}$ lie in $N(q)$ and are inextendible therein. But since these do not intersect $S$, we contradict the global hyperbolicity of $N(q)$. \\ \indent 
It now follows that $\text{edge}(\phi(S))$, if not empty, lies outside of $\phi(S)$. By standard results in causal theory, $H(\phi(S))$ is ruled by null geodesics intersecting $\text{edge}(\phi(S))$. By the definition of $D_{\epsilon}(S)$, we can use $\phi$ to pull back $H(\phi(S))\cap \phi(D_{\epsilon}(S))$ into $D_{\epsilon}(S)\cap M$. By the definition of $D_{\epsilon}(S)$, it is then clear that $\phi^{-1}\left[ H(\phi(S))\cap \phi(D_{\epsilon}(S)\right] \cap D(S)\neq \emptyset$, and thus $(M',g')$ is holed. 
\end{proof}

\end{document}